\documentclass[11pt]{article}
\usepackage{a4wide}
\usepackage{amssymb,amsmath,amsthm,amsfonts,mathtools}
\usepackage{graphicx,color,tikz}
\usepackage[footnotesize]{caption}
\usepackage{natbib}
\usepackage{dsfont}

\newtheorem{theorem}{Theorem}

\newtheorem{definition}{Definition}
\newtheorem{remark}{Remark}
\newtheorem{corollary}{Corollary}
\DeclareMathOperator*{\argmax}{arg\,max}
\DeclareMathOperator*{\argmin}{arg\,min}
\begin{document}
	\baselineskip=1.525\baselineskip
	
	\title{Conditional strategy equilibrium\footnote{This paper is accepted at the 4th Games, Agents, and Incentives Workshop (GAIW 2022). Held as part of the Workshops at the 20th International Conference on Autonomous Agents and Multiagent Systems. We thank three anonymous referees for their valuable comments.}}
	
	\author{Lorenzo Bastianello\thanks{LEMMA, Universit\'e Paris 2 Panth\'eon-Assas, Paris, France. E-mail: lorenzo.bastianello@u-paris2.fr} \and Mehmet S. Ismail\thanks{Department of Political Economy, King's College London, UK. E-mail: mehmet.s.ismail@gmail.com.}}
	\date{\today \\ First version: March 2021}

	\maketitle
	
	\begin{abstract}
	In this note, we prove the existence of an equilibrium concept, dubbed conditional strategy equilibrium, for non-cooperative games in which a strategy of a player is a function from the other players' actions to her own actions. We study the properties of efficiency and coalition-proofness of the conditional strategy equilibrium in $n$-person games. 
\end{abstract}
	
	\section{Introduction} \label{sec:intro}

In introductory  game theory courses, the two classes of games that are  first introduced to the audience are  normal form games and extensive form games, see for instance \citet{maschler2013}. It is immediate to note the different definitions given to the concept of ``strategy.'' In normal form games a strategy for player $i$ is simply an action $a_i$ chosen from a set $A_i$. In extensive form games, strategies are more complicated objects as they specify an action for each possible node (or information set) of the game. Therefore, in this latter case, a player can condition her action on what happened previously in the game. 

Our starting point is the observation that even in  normal form games player $i$ may want to condition on player $j$'s actions. Consider the following illustrative example. Two players would like to take some vacations. Player $i$ may formulate such sentences:
\begin{itemize}
	\item ``I go on vacation \textit{if you go}.''
	\item ``I don't go on vacation \textit{if you don't go}.''
\end{itemize} 
Player $i$ is conditioning her action on the actions of player $j$, and of course the same reasoning can be done with $i$ and $j$ inverted. 

In this paper, we develop a generalization of the concept of strategy for normal form games in order to capture the intuition given in the example described above. More formally, let  $A_i$ be the set of available actions of Player $i$. We define a \textit{(pure) conditional strategy} as a function $s_i:A_{-i}\rightarrow A_i$. Hence, a conditional strategy tells which action player $i$ would play if the other players play the profile $a_{-i}\in A_{-i}$. A \textit{conditional extension} of a game is the game in which players use their conditional strategies.

Note, however, that when players use conditional strategies, it is not always clear which profile of actions will be selected, and therefore which payoffs the players will receive. Consider the the example above in which two players wanted to go on a vacation. The action set of player 1 is $A_1=\{Go,\neg Go\}$ (i.e. ``Go''  and ``don't Go'' on vacation) and the one of player 2 is $A_2=\{go,\neg go\}$. Consider the following conditional strategies $s_1:\{go,\neg go \}\rightarrow \{Go,\neg Go\}$ and $s_2:\{Go,\neg Go\}\rightarrow \{go,\neg go \}$ defined in the table below
\begin{center}
	\begingroup
	\setlength{\tabcolsep}{10pt} 
	\renewcommand{\arraystretch}{1.5} 
	\begin{tabular}{ l | l  }
		\multicolumn{1}{c}{$s_1$} & \multicolumn{1}{c}{$s_2$} \\ \hline
		$s_1(go)=Go$ & $s_2(Go)=\neg go$\\ 
		$s_1(\neg go)=\neg Go$  & $s_2(\neg Go)=go$
	\end{tabular}
	\endgroup
\end{center}
These conditional strategies create a circularity and do not pinpoint an outcome. Another problem may be that some conditional strategy profiles pinpoint multiple outcomes. Consider for instance the strategies defined  by  $s_1(go)=Go$, $s_1(\neg go)=\neg Go$ and $s_2(Go)=go$, $s_2(\neg Go)=\neg go$.\footnote{Note however that agreements can be made if one profile is unambiguously better than all others (Pareto dominates).} We call these situations \textit{disagreements}.

We solve these difficulties by extending the utility from profiles of actions (given by the original normal form game) to profiles of conditional strategies using general extension functions.\footnote{Note that the disagreement phenomenon is not unique to games in conditional strategies. A special case of a disagreement occurs with completely mixed strategies. Every completely mixed strategy profile creates a disagreement because it does not pinpoint a unique outcome. Conventionally, we extend the Bernoulli utility function to the (von Neumann-Morgenstern expected utility function in these situations. However, one could use a different extension function such as maxmin expected utility of \citet{gilboa1989}.} In the benchmark model, we assume that a conditional strategy profile that produces a disagreement induces the worst utility in the game to the players, which is similar to the disagreement outcome in the Nash bargaining model \cite{nash1953}; though, we explore more general alternatives in Section~\ref{sec:general}.

There are several interpretations available in the literature regarding conditional commitments. A straight-forward interpretation is that players submit simultaneously conditional strategies to a computer (see, e.g., Tennenholtz's \cite{tennenholtz2004} program equilibrium). Another interpretation is that players submit their strategies to a trusted third party; though, players may not need a third party if there is a common understanding of what to do when their conditional choices create a disagreement (e.g., tossing a coin to decide on the outcome of a disagreement).

We naturally define a conditional strategy equilibrium as a Nash equilibrium \cite{nash1950a} of the game played in conditional strategies. Our first result, Theorem \ref{thm:existence_pure}, shows that a pure conditional strategy equilibrium exists in every $n$-person game. The generalization of the concept of action to the one of conditional strategy allows us to prove several important results about Pareto efficiency of conditional strategy equilibria. Our Folk Theorem (Theorem \ref{thm:folk}) shows that in two-person games, any payoff above the maxmin payoff can be supported as a pure conditional equilibrium. This implies for instance that cooperation can be achieved as a conditional strategy equilibrium in the prisoner's dilemma. In three-person games, Theorem \ref{thm:Pareto_3pl} shows that there exists a Pareto optimal conditional strategy equilibrium. 

Suppose that in a game (i) the disagreement payoff is the worst payoff, and (ii) the number of players is four or more. Then, it is easy to show that every action profile can be supported as a conditional strategy equilibrium. This happens because it is possible to define a conditional strategy profile in which if a player unilaterally deviates, players end up in a disagreement situation in which everyone gets their minimal payoff. Therefore no player will want to deviate unilaterally. This shows the need to consider a stronger notion of equilibrium for $n$-person games. In Section~\ref{sec:n-person}, we introduce the strong conditional equilibrium, which is a natural extension of strong Nash equilibrium \cite{aumann1959} to conditional strategies. It is well known that a strong Nash equilibrium almost never exists. In Theorem~\ref{thm:strongCE}, we give an easy-to-check sufficient condition for its existence when conditional strategies are used.

One of the main advantages of going from pure strategies to mixed extension is that there is in general no Nash equilibrium in pure strategies but there is always one in mixed strategies. However, this comes with a computational cost in part because the set of mixed strategies is an infinite set and the existence theorem relies on a non-constructive fixed point theorem. By contrast, while conditional extension does increase the number of strategies of each player, the cardinality of the set of strategies has a finite bound. All of our existence results (i.e., theorems 1$-$5) except Theorem~\ref{thm:conditional_mixed_extension} are constructive.

\subsection{A brief literature review}\label{sec:literature}

This section gives a brief overview of the previous contributions to the idea of conditioning. We do not attempt to give a comprehensive review of the related literature in this short note.

Conditioning players' actions on the history of the game is a concept that is well known in extensive form games. In normal form games, the idea that players may be better off conditioning their strategies on opponents' actions dates back at least to \citet{schelling1956essay}. \citet{schelling1956essay} discussed bargaining and advanced the idea that  irreversible commitment to a certain strategy may serve as a threat and may  benefit the player: ``[a player] must commit himself to a \textit{conditional} choice'' (the emphasis is present in the original paper). However, \citet{schelling1956essay} does not entirely clarify whether the game is played sequentially or simultaneously. In the context of irreversible commitments, \citet{brams1994} distinguishes between compellent and deterrent threat.

\citet{howard1971} first studied conditional commitments in the context of prisoners' dilemma and other games. In the book of \citet{howard1971} only one player at a time can condition on the other action, while the other player can only play a fixed action (in the terminology of this paper, a constant conditional strategy). He calls these games metagames.

In a seminal work, \citet{tennenholtz2004} introduced the notion of program equilibrium. The main idea is that players write computer programs that in turn select a strategy.  Strategies are selected by conditioning on other players' programs. In this setting \citet{tennenholtz2004} shows that cooperation can be achieved in the prisoner's dilemma. This approach has been extended by \citet*{kalai2010}. Instead of computer programs, players can select general conditional commitment devices. \citet{kalai2010} prove a folk theorem in this setting. The main difference between \citet{tennenholtz2004} and \citet{kalai2010} is that Kalai et al. use conditioning devices in order to overcome the problems of circularities and multiplicities explained in the Introduction.

The rest of the paper formally proves the results presented in the Introduction.  Section \ref{sec:definition} gives the definition of pure conditional strategies. Section \ref{sec:results} shows how to reach Pareto efficiency in two and three-person games. Section \ref{sec:n-person} studies the $n$-person case. Section \ref{sec:general} considers a more general way to solve possible disagreements induced by conditional strategy profiles. Finally Section \ref{sec:mixed} introduces mixed strategies in conditional games.

\section{Definition of pure conditional strategies}\label{sec:definition}

Let $(A_i, u_i)_{i\in N}$ be an $n$-person noncooperative game, where $N=\{1,...,n\}$ is the finite set of players, $A_i=\{a^1_i, a^2_i, ..., a^{m_i}_i \}$ finite non-empty action set, and $u_i:A\rightarrow \mathbb{R}_+$ the (von Neumann-Morgenstern) utility function of player $i\in N$. A strategy profile is denoted by $a\in A = \times_{i\in N}A_i$.\footnote{Without loss of generality, we consider non-negative utility function.}

A pure conditional strategy of player $i$ is a function $s_i:A_{-i}\rightarrow A_i$. A pure conditional strategy profile is $s\in S=\times_{i\in N} S_{i}$ where $S_i$ is the set of conditional strategies of player $i$. 

\begin{remark}
	The number of pure conditional strategies of player $i$ is given by $m_i^{\prod_{j\neq i} m_j}=|S_i|$ because $\prod_{j\neq i} m_j$ gives the number of pure action profiles of everyone but $i$.
\end{remark}

A pure conditional strategy profile $s\in S$ is an \textit{agreement} if there exists a fixed point $a$ of $s$---i.e., for all $i$ $s_i(a_{-i})=a_i$---such that for any  fixed point $a'$, $u_i(a)\geq u_i(a')$. The intuition is that everybody agrees on the unique strategy profile $a$. If $s\in S$ is not an agreement, then it is called a \textit{disagreement}.

Let $B=\{s\in S| s~\text{is an agreement}\}$ be the set of agreements. The utility function $U_i:S\rightarrow \mathbb{R_+}$ is defined as follows. 
\[
U_i(s) = 
\begin{cases}
u_i(a) & \text{if}~s\in B \\
0 & \text{if}~s\in S\setminus B,
\end{cases}
\] 
where $a=s(a)$. This is a well-defined function since every $s\in B$ has a unique fixed point, and the disagreement payoff is zero. Note that $U_i|_{B}=u_i$.

Let $(S_i, U_i)_{i\in N}$ denote a game in \textit{conditional extension} of the game $(A_i, u_i)_{i\in N}$. 

\begin{remark}
	\label{rem:fixed_point_existence}
	In every conditional extension, agreements or fixed points of strategy profiles exist because $A_i$ is nonempty for each $i$.
\end{remark}

We next define conditional strategy equilibrium.

\begin{definition}
	A conditional strategy profile $s^*\in S$ is called a \textit{conditional strategy equilibrium} if for all $i$ $s^*_i\in \argmax U_i(s^*)$ or equivalently
	\[
	U_i(s^*)\geq U_i(s'_i, s^*_{-i}).
	\]
\end{definition}

In other words, a conditional strategy equilibrium (CSE) is a self-enforcing agreement in the space of conditional strategy profiles, just like a Nash equilibrium \citep{nash1950a} is a self-enforcing agreement in the space of mixed strategy profiles.

\section{Results}
\label{sec:results}

\begin{theorem}[Existence]
	\label{thm:existence_pure}
	In every $n$-person game $G=(S_i, U_i)_{i\in N}$, there exists a pure conditional equilibrium.
\end{theorem}

\begin{proof}
	Let $BR_i$ denote the best-response correspondence of player $i$, i.e., $BR_i(s_{-i})=\{s_{i}\in S_{i} | s_i\in \argmax_{s'_{i}\in \S_i} U_i(s'_{i}, s_{-i})\}.$ Let $\bar{S}_i$ denote the set of constant strategies of player $i$, i.e., $\bar{s}_i\in \bar{S}_i$ if and only if for every $a_{-i}$ $\bar{s}_i(a_{-i})=a_i$.
	
	Player 1 chooses $\hat{s}_1$ such that $\hat{s}_1(a_{-1})\in BR_1(a_{-1})$ for every $a_{-1}\in A_{-1}$. 
	Let $$\hat{a}_2(a_3,...,a_n)\in \argmax_{a'_{2}\in A_2} u_2(\hat{s}_1(a'_2,a_3,...,a_n),a'_{2}, a_3,...,a_n).$$ Then, define the conditional strategy of player 2 as $\hat{s}_{2}(a_{-2})=\hat{a}_2(a_3,...,a_n)$ for every $a_{-2}$. In plain words, fixing the action profile $(a_3,...,a_n)$ player 2 chooses an action that maximizes her utility given that player 1 best responds to her action.
	
	Analogously, consider $\hat{a}_3(a_4,...,a_n)$ in  
	\[\argmax_{a'_{3}\in A_3} u_3(\underbrace{\hat{s}_1(\underbrace{\hat{a}_2(a'_3,...,a_n)}_\text{$\hat{a}_2$},a'_3,a_4,...,a_n)}_\text{$\hat{a}_1$},\underbrace{\hat{a}_2(a'_3,...,a_n)}_\text{$\hat{a}_2$}, a'_3, a_4,...,a_n).
	\]
	Then, define the conditional strategy of player 3 as $\hat{s}_{3}(a_{-3})=\hat{a}_3(a_4,...,a_n)$ for every $a_{-3}$. Conditional strategy of player $k$ is defined as $\hat{s}_{k}(a_{-k})=\hat{a}_k(a_{(k+1)},...,a_n)$ for every $a_{-k}$ where $\hat{a}_k$ is defined analogously. For player $n$, $$\hat{a}_n\in \argmax_{a'_{n}\in A_n} u_n(\hat{a}_1,..., \hat{a}_{n-1},a'_n),$$ and $\hat{s}(a_{-n})=\hat{a}_n$. Define with a slight abuse of notation $\hat{a}=(\hat{a}_1,\hat{a}_2,...,\hat{a}_n)$ where $\hat{a}_i=\hat{a}_i(\hat{a}_{i+1},...,\hat{a}_n)$
	
	Next, we show that $\hat{s}\in B$, i.e., it is an agreement. We will prove that $s(\hat{a})=\hat{a}$ and that for all $a$ such that $\hat{s}(a)=a$, then $a=\hat{a}$. This is because 
	$\hat{s}_n(a_{-n})=\hat{a}_n$ by definition of $\hat{s}_n$, so $a_n=\hat{a}_n$. Then, given $\hat{a}_n$, $\hat{s}_{n-1}(a_{-(n-1)})=\hat{a}_{n-1}(a_n)=\hat{a}_{n-1}(\hat{a}_n)=\hat{a}_{n-1}$ by definition of $\hat{s}_n$, so $a_{n-1}=\hat{a}_{n-1}$. We repeat this process till we obtain $a_2=\hat{a}_2$. Then, notice that $\hat{s}_{1}(a_{-1})=BR_1(a_{-1})=BR_1(\hat{a}_{-1})=\hat{a}_1$.
	
	Now we show that $\hat{s}$ is a conditional equilibrium. First, we prove that player $n$ has no unilateral profitable deviation from $\hat{s}$. Suppose by way of contradiction that $\tilde{s}_n$ is such a deviation. Consider $\tilde{s}=(\tilde{s}_n,\hat{s}_{-n})$, which must be an agreement in $B$ because otherwise $U_n(\tilde{s})=0$. Then, there must be a unique $\tilde{a}\in A$ such that $\tilde{s}(\tilde{a})=\tilde{a}$. Let $\tilde{a}_n=\tilde{s}_n(\tilde{a}_{-n})$. Then, $u_n(\tilde{a}_n,\hat{a}_{-n})\leq u_n(\hat{a})=\max_{a'_{n}\in A_n} u_n(a'_n,\hat{a}_{-n})$. Thus, $\tilde{s}_n$ is not a unilateral profitable deviation. For every player $k>1$, the reasoning is analogous. Finally, we show that player $1$ has no unilateral profitable deviation. Fixing $\hat{s}_{-1}$ induces the action profile $\hat{a}_{-1}$. Since $\hat{s}_{1}$ is by definition a best reply to $\hat{a}_{-1}$, player 1 cannot strictly benefit from a unilateral deviation.
\end{proof}

\begin{theorem}[$2$-player Folk theorem]
	\label{thm:folk}
	Let $G=(S_i, U_i)_{i\in \{1,2\}}$ be a two-person game and $\bar{a}\in A$ be an action profile. Players each receive at least their individually rational (maximin) payoff at $\bar{a}$ if and only if it can be supported as a conditional strategy equilibrium outcome.
\end{theorem}

\begin{proof}
	First, we show that if at some profile $\bar{a}$ players each receive at least their individually rational payoff, then we can construct a conditional strategy equilibrium $s$ in which the unique fixed point of the profile is $\bar{a}$, i.e., $s_1(\bar{a}_2)=\bar{a}_1$ and $s_2(\bar{a}_1)=\bar{a}_2$. We define $s$ as follows. For every $i$ let $s_i(\bar{a}_{-i})=\bar{a}_i$. For every $i$ and every $a_{-i}\neq \bar{a}_{-i}$ define $s_i(a_{-i})$ such that $s_i(a_{-i})\in \argmin_{a'_{i}\in A_i}u_{-i}(a'_i,a_{-i})$. 
	
	We next show that $s$ is an agreement. By construction, we have $s(\bar{a})=\bar{a}$. Take $\hat{a}\neq \bar{a}$ such that $s(\hat{a})=\hat{a}$. Then, we have $\hat{a}_1\in \argmin_{a'_{1}\in A_1}u_2(a'_1,\hat{a}_{2})$ and $\hat{a}_2\in \argmin_{a'_{2}\in A_2}u_1(\hat{a}_{1}, a'_2)$. 
	Since at $\bar{a}$ each player receives their individually rational payoff, we have for every $a'_i\neq \bar{a}_i$ we have $u_i(\bar{a})\geq \min_{a_{-i}\in A_{-i}}u_{i}(a'_i,a_{-i})$ and in particular for $a'_i\neq \hat{a}_i$, $u_i(\bar{a})\geq \min_{a_{-i}\in A_{-i}}u_{i}(\hat{a}_i,a_{-i})=u_i(\hat{a})$. Therefore $\bar{a}$ Pareto dominates $\hat{a}$, so $s$ is an agreement.
	
	Next, we show that there is no unilateral profitable deviation from $s$. To reach a contradiction, suppose that for some $i$ $s'_i$ is a unilateral profitable deviation from $s$. If $(s'_i, s_{-i})$ is not an agreement, then $U_i(s'_i, s_{-i})=0$; so $s'_i$ is not a unilateral profitable deviation. Suppose now that $(s'_i, s_{-i})$ is an agreement with  fixed  point $a'$. Then by individual rationality and definition of $s_{-i}$, $u_i(\bar{a})\geq u_{i}(a'_i,s_{-i}(a'_i))$ so $s'_i$ is not a unilateral profitable deviation. This completes the proof that $s$ is a pure conditional strategy equilibrium.
	
	Second, we show that if a profile $s$ is a conditional strategy equilibrium, then at $s$ players each receive at least their individually rational payoff. To reach a contradiction, suppose that there is a player $i$ such that $U_i(s) < \max_{a_{i}\in A_{i}}\min_{a_{-i}\in A_{-i}}u_{i}(a_i,a_{-i})$. Then, player $i$ has a profitable deviation to the constant conditional strategy $s'_i$ defined for every $a'_{-i}\in A_{-i}$ as  $s'_i(a'_{-i})=a_i''$ where $a_i''\in \argmax_{a_{i}\in A_{i}}\min_{a_{-i}\in A_{-i}}u_{i}(a_i,a_{-i})$. Note that  profile $(s'_i,s_{-i})$ defines the unique fixed point $(a_i'',s_{-i}(a_i''))$. Then $s'_i$ is a profitable deviation because $$U(a_i'',s_{-i}(a_i''))\geq \max_{a_{i}\in A_{i}}\min_{a_{-i}\in A_{-i}}u_{i}(a_i,a_{-i})>U_i(s).$$
\end{proof}

The next theorem shows that Pareto optimality of conditional strategy equilibrium can be achieved in not only two-person games but also three-person games.

\begin{theorem}[$3$-player Pareto]
	\label{thm:Pareto_3pl}
	In every three-person game $G=(S_i, U_i)_{i\in \{1,2,3\}}$ there exists a Pareto optimal conditional strategy equilibrium.
\end{theorem}

\begin{proof}
	Let $G$ be a three-person game and $\bar{a}$ be a Pareto optimal action profile such that for some player $i$ $u_i(\bar{a})$ is the maximum payoff of $i$ in $G$. Note that $\bar{a}$ exists in every finite three-person game. 
	
	We next construct a conditional strategy equilibrium $s$ in which the unique fixed point of the profile is $\bar{a}$. For every player $l$, define $s_l(\bar{a}_{-l})=\bar{a}_l$. For the rest of the proof  fix an action profile $a'\in A$ such that for every $l$, $a'_l\neq \bar{a}_l$. Define 
	
	\begin{equation}
	\label{eq:1}
	s_i(a_k,\bar{a}_j)=a'_i \text{ for all } a_k\neq \bar{a}_k, 
	\end{equation}
	\begin{equation}
	\label{eq:2}
	s_i(\bar{a}_k,a_j)=a'_i  \text{ for all }  a_j\neq \bar{a}_j,
	\end{equation}
	\begin{equation}
	\label{eq:3}
	s_i(a_k,a_j)=\bar{a}_i \text{ for all } a_k\neq \bar{a}_k \text{ and } a_j\neq \bar{a}_j,
	\end{equation}
	\begin{equation}
	\label{eq:4}
	s_k(\bar{a}_i,a_j)=\bar{a}_k \text{ for all } a_j\neq \bar{a}_j,
	\end{equation}
	\begin{equation}
	\label{eq:5}
	s_k(a_i,\bar{a}_j)=\bar{a}_k \text{ for all } a_i\neq \bar{a}_i,
	\end{equation}
	\begin{equation}
	\label{eq:6}
	s_k(a_i,a_j)=a'_k \text{ for all } a_i\neq \bar{a}_i \text{ and } a_j\neq \bar{a}_j,
	\end{equation}
	\begin{equation}
	\label{eq:7}
	s_j(\bar{a}_i,a_k)=\bar{a}_j \text{ for all } a_k\neq \bar{a}_k,
	\end{equation}
	\begin{equation}
	\label{eq:8}
	s_j(a_i,\bar{a}_k)=\bar{a}_j \text{ for all } a_i\neq \bar{a}_i,
	\end{equation}
	\begin{equation}
	\label{eq:9}
	s_j(a_i,a_k)=a'_j \text{ for all } a_i\neq \bar{a}_i \text{ and } a_k\neq \bar{a}_k.
	\end{equation}
	
	First notice that $\bar{a}$ is the unique fixed point point of $s$ by construction. We show that no player has a unilateral profitable deviation from $s$. Clearly, $s_i$ is already a best response to $s_{-i}$ because player $i$ receives the highest payoff in the game. 
	
	Note that if for every $a\in A\setminus {\bar{a}}$, the followings hold
	\begin{equation}
	\label{eq:10}
	(s_i(a_k,a_j),s_j(a_i,a_k))\neq (a_i,a_j)
	\end{equation}
	\begin{equation}
	\label{eq:11}
	(s_i(a_k,a_j),s_k(a_i,a_j))\neq (a_i,a_k)
	\end{equation}
	then neither player $k\neq i$ nor player $j\neq i$ can unilaterally deviate from $s$ and create an agreement. Hence if they deviate they would both get 0. To see this, consider e.g. expression (\ref{eq:11}) and suppose that $(s_i(a_k,a_j),s_j(a_i,a_k))= (a_i,a_j)$. Then player $k$ could create an agreement by deviating to $s_k'(a_i,a_j)=a_k$. Therefore we only need to show expressions (\ref{eq:10}) and (\ref{eq:11}).
	
	Fix $a\in A$, $a\neq \bar{a}$. There are two cases to consider: (1) $s_i(a_k,a_j)=\bar{a}_i$; and (2) $s_i(a_k,a_j)=a'_i$.
	
	Case 1.1: $s_i(a_k,a_{j})=\bar{a}_i$ and $a_i=\bar{a}_i$. Then $a_k\neq \bar{a}_k$ and $a_j\neq \bar{a}_j$ by Equation~\ref{eq:3}. But then $s_k(a_{-k})=\bar{a}_k$ and $s_j(a_{-j})=\bar{a}_j$ by Equation~\ref{eq:4} and~\ref{eq:7}, respectively. Thus, both expressions~\ref{eq:10} and~\ref{eq:11} hold true.
	
	Case 1.2:  $s_i(a_k,a_{j})=\bar{a}_i$ and $a_i=a'_i$. Then both expressions~\ref{eq:10} and~\ref{eq:11}  hold true because $s_i(a_{-i})=\bar{a}_i$ but $a_i\neq \bar{a}_i$.
	
	Case 2.0: $a_i\neq a'_i$. Then both expressions~\ref{eq:10} and~\ref{eq:11} hold true because $s_i(a_k,a_j)=a'_i$.
	
	Case 2.1:  $s_i(a_k,\bar{a}_{j})=a'_i$ and $a_i=a'_i$. Then by Equation~\ref{eq:1} $a_k\neq \bar{a}_k$ $a_j=\bar{a}_{j}$. Equation~\ref{eq:5} implies that $s_k(a_i,\bar{a}_{j})=\bar{a}_k$ and Equation~\ref{eq:9} implies that $s_j(a_i,a_k)=a'_j$. But $a_j=\bar{a}_{j}\neq a'_j$, so Expression~\ref{eq:10} holds; and $a_k\neq \bar{a}_k$, so Expression~\ref{eq:11} holds too.
	
	Case 2.2:  $s_i(\bar{a}_k,a_{j})=a'_i$ and $a_i=a'_i$. Then  by Equation~\ref{eq:2} $a_j\neq \bar{a}_j$ and $a_k=\bar{a}_{k}$. Equation~\ref{eq:8} implies that $s_j(a_i,\bar{a}_k)=\bar{a}_j$ and Equation~\ref{eq:6} implies that $s_k(a_i,a_j)=a'_k$. But $a_j\neq \bar{a}_j$, so Equation~\ref{eq:10} is satisfied; and $a_k=\bar{a}_{k}\neq \bar{a}_k$, so Equation~\ref{eq:11} is also satisfied.
\end{proof}

\section{Strong coalition-proofness in $n$-person games}\label{sec:n-person}

Let $G$ be a $n$-person game with $n\geq 4$. It is possible to show that every action profile can be supported as a conditional strategy equilibrium when we assume that a disagreement gives the worst payoff in the game. The idea is to define a conditional strategy profile with the property that no player can unilaterally deviate and create a fixed point. Hence, any deviation will result in the player getting 0. This motivates us to introduce the strong conditional equilibrium, which is a natural extension of strong Nash equilibrium to conditional strategies.

\begin{definition}
	A conditional strategy profile $s^*\in S$ is called a \textit{strong conditional equilibrium} if there is no subset $C\subseteq N$, $C\neq\emptyset$, and $s_C\in \prod_{j\in C} S_j$ such that for all $i\in C$
	\[
	U_i(s_C,s^*_{-C})> U_i(s^*).
	\]
\end{definition}

It is well known that a strong Nash equilibrium almost never exists. While Figure~\ref{fig:strong_nonexistence} shows that strong conditional equilibrium does not exist in general either,  we give below an easy-to-check sufficient condition for the existence of strong conditional equilibrium when conditional strategies are used.

\begin{figure}[h]
	\begin{minipage}[b]{\linewidth}\centering
		\[
		\begin{array}{ r|c|c| }
		\multicolumn{1}{r}{}
		&  \multicolumn{1}{c}{A}
		& \multicolumn{1}{c}{B}\\
		\cline{2-3}
		x&  2,1,0 & 0,2,1 \\
		\cline{2-3}
		y&  0,0,0 & 0,2,1 \\
		\cline{2-3}
		\end{array}
		\qquad
		\begin{array}{ r|c|c| }
		\multicolumn{1}{r}{}
		&  \multicolumn{1}{c}{A}
		& \multicolumn{1}{c}{B}\\
		\cline{2-3}
		x&  2,1,0 & 0,0,0 \\
		\cline{2-3}
		y&  1,0,2 & 1,0,2 \\
		\cline{2-3}
		\end{array}
		\]
	\end{minipage}
	\caption{A three player game which does not admit a strong conditional equilibrium. Player 3 chooses between the matrices $L$ (left) and $R$ (right).}
	\label{fig:strong_nonexistence}
\end{figure}

\begin{remark}
	\label{rem:strong_nonexistence}
	Figure~\ref{fig:strong_nonexistence} illustrates a three-player  counterexample in which there is no strong conditional equilibrium.
\end{remark}

To see this, note the following profitable deviations from every action profile. Players 2 and 3 can  profitably (and jointly) deviate to $(x,B,L)$ from profiles $(x,A,R)$ and $(x,A,L)$. Similarly, from $(x,B,L)$ and $(y,B,L)$, players 1 and 3 can profitably deviate to $(y,B,R)$. From $(y,A,R)$ and $(y,B,R)$, players 1 and 2 can profitably deviate to $(x,A,L)$ and $(x,A,R)$. And, it is clear that player 1 and player 2 can profitable deviate from profiles $(x,B,R)$ and $(y,A,L)$, respectively. For any action profile $a$, these deviations create a unique agreement, so they are profitable deviations from any conditional strategy profile whose unique fixed point is $a$. As a result, there is no strong conditional equilibrium in this game.

\begin{corollary}[2-player existence]
	\label{cor:strong2person}
	Let  $G=(S_i, U_i)_{i\in \{1,2\}}$ be a two-person game. Then, there exists a strong conditional equilibrium in pure strategies.
\end{corollary}

This corollary directly follows from the Theorem~\ref{thm:folk} because every Pareto optimal conditional equilibrium is a strong conditional equilibrium in two-person games.

\begin{theorem}[$n$-player sufficient condition]
	\label{thm:strongCE}
	Let $(S_i, U_i)_{i\in N}$ be an $n$-person game. If there is an action profile $\bar{a}$ where at least two players receive their highest payoffs in the game, then there exists a strong conditional equilibrium in pure strategies that supports  $\bar{a}$.
\end{theorem}

\begin{proof}
	Let $\bar{a}$ be an action profile in which players $i$ and $j$ get their maximum payoffs. We construct a strong conditional equilibrium $s^*$ whose unique fixed point is $\bar{a}$. For every player $m$, define $s^*_m(\bar{a}_{-m})=\bar{a}_m$. Let $s_m^*(a_{-m})\in A_m$ for any $m$ with the following constraint for $i$ and $j$. For some  $l\in\{i,j\}$ fix $a_{-l}\in A_{-l}$. If $s^*_l(a_{-l})=a_l'$ for some $a'_l\in A_l$ such that $(a_l',a_k,a_{-kl})\neq \bar{a}$, then  $s^*_k(a_l',a_{-kl})\neq a_k$ where $k\in\{i,j\}$ and $k\neq l$. It is always possible to construct such $s^*$ because both $i$ and $j$ have at least two different actions. 
	
	First, notice that $s^*$ has a unique fixed point which is $\bar{a}$. Second, neither player $i$ nor player $j$ has a unilateral profitable deviation because they each receive their maximum payoff in the game at $s^*$. By the same token, neither $i$ nor $j$ would join any coalition to deviate from $s^*$. 
	
	In addition, for every coalition $C\subseteq N\setminus \{i,j\}$, $C\neq\emptyset$, and every $s\neq s^*$ we have $U_i(s_C,s^*_{-C})=0$. This is because there is no $a\in A$ such that $(s_C,s^*_{-C})(a)= a$. To see this, suppose that there exists $a\in A$ such that $(s_C,s^*_{-C})(a)= a$. Then  $s^*_i(a_{-i})=a_i$ and $s^*_j(a_{-j})=a_j$. However this contradicts  the constraint above. In plain words, the players excluding $i$ and $j$  cannot create another fixed point. Thus, there is neither a unilateral nor a joint profitable deviation from $s^*$, which implies that $s^*$ is a strong conditional equilibrium.
\end{proof}

\section{Extension to general utility functions}
\label{sec:general}

In this section we generalize the definition of function $U$ given in Section \ref{sec:definition} to cases in which there is no agreement.

Let $f_i:S\rightarrow \mathbb{R}$ be a function. Define the (extended) utility function $\tilde{U}_i:S\rightarrow \mathbb{R}$ as follows. 
\[
\tilde{U}_i(s) = 
\begin{cases}
u_i(a) & \text{if}~s\in B \\
f_i(s) & \text{if}~s\in S\setminus B,
\end{cases}
\] 
where $a=s(a)$.

Let $(S_i, \tilde{U}_i)_{i\in N}$ denote the conditional extension of the game $(A_i, u_i)_{i\in N}$ where $\tilde{U}_i$ extends $u_i$ as above.

In section~\ref{sec:definition}, we assumed that the extended utility function assigns the worst payoff to disagreements. Next, we explore the case in which if a conditional strategy profile $s$ is a disagreement, then player $i$ gets the average of payoffs. For every $s\not\in B$  define the set 
\[
D(s) = 
\begin{cases}
\cup_{i\in N} \cup_{a_{-i}\in A_{-i}} \{(s_i(a_{-i}),a_{-i})\} & \text{if}~s\text{ has no fixed point}  \\
\{a\in A | s(a)=a \} & \text{otherwise}.
\end{cases}
\] 
Define $f_i(s):=\frac{1}{|D(s)|}\sum_{a\in D(s)} u_i(a)$. Let $(S_i, \pi_i)_{i\in N}$ denote the game in conditional extension where players each receive their average payoff from a disagreement.

\begin{theorem}[Two-person existence]
	\label{thm:general}
	Every two-person game $(S_i, \pi_i)_{i\in \{1,2\}}$ has a conditional strategy equilibrium in pure strategies.
\end{theorem}

\begin{proof}
	First, suppose that $G$ has a pure Nash equilibrium $a^*$. Then, conditional strategy profile $s$ in which each player plays the constant strategy $s_i(a_{-i})=a_{i}^*$ for all $a_{-i}\in A_{-i}$ is a conditional strategy equilibrium. This is because if player $i$ unilaterally deviates to $s_i'$ in which  $s_i'(a_{-i}^*)=a_i'\neq a_i^*$ then the unique fixed point of $(s_i',s_{-i})$ is $(a'_i,a_{-i}^*)$. Thus $i$ would get $u_i(a'_i,a_{-i}^*)\leq u_i(a^*)$.

	Second, suppose that $G$ does not have a pure Nash equilibrium. Then, consider conditional strategy profile $\bar{s}$ in which for every $i$ and every $a_{-i}$ $\bar{s}_i(a_{-i})\in BR_i(a_{-i})$. It is clear that $\bar{s}$ has no fixed point because otherwise the fixed point would be a pure Nash equilibrium contradicting our supposition.
	
	There are two cases to consider: (i) $\bar{s}$ is a conditional strategy equilibrium; (ii) there is at least a player $i$ who can unilaterally deviate from $\bar{s}$ to some best response $s_i$ and increase their payoff, i.e. $\pi_i(s_i,\bar{s}_{-i})>\pi_i(\bar{s})$. We only need to consider case (ii). Next we construct a conditional strategy equilibrium $s^*=(s_i^*,\bar{s}_{-i})$ in which $s_i^*$ is a constant conditional strategy.
	
	First, suppose that $(s_i,\bar{s}_{-i})$ has no fixed points. Then for all  $a_{-i}\in A_{-i}$ $u_i(s_i(a_{-i}),a_{-i})\leq u_i(\bar{s})$ since $\bar{s}_i(a_{-i})\in BR_i(a_{-i})$. Note that $|D(s_i,\bar{s}_{-i})|=|D(\bar{s})|$ because both profiles have no fixed points, therefore 
	$$
	\pi_i(s_i,\bar{s}_{-i})=\frac{1}{|D(s_i,\bar{s}_{-i})|}\sum_{a\in D(s_i,\bar{s}_{-i})} u_i(a)$$
	$$\leq \frac{1}{D(\bar{s})} \sum_{a\in D(\bar{s})} u_i(a) = \pi_i(\bar{s})
	$$ 
	so $s_i$ is not a profitable deviation for player $i$. Hence, for $s_i$ to be a profitable deviation, $(s_i,\bar{s}_{-i})$ must have one or more  fixed points.
	
	Second, suppose there are multiple fixed points. Then since $|D(s_i,\bar{s}_{-i})|$ is finite, there is $a'\in A$ such that $u_i(a')\geq u_i(a)$ for all $a\in D(s_i,\bar{s}_{-i})$. Hence 
	\begin{equation}\label{eq:general-f-2}
	\frac{1}{|D(s_i,\bar{s}_{-i})|}\sum_{a\in D(s_i,\bar{s}_{-i})} u_i(a)\leq u_i(a').
	\end{equation}
	
	Now define $s_i^*(a_{-i})=a'_i$ for all $a_{-i}\in A_{-i}$. By (\ref{eq:general-f-2}), since $s_i$ is a best response to $\bar{s}_{-i}$, $s_i^*$ is also a best response to $\bar{s}_{-i}$.

	Finally we show that $\bar{s}_{-i}$ is a best response to constant conditional strategy $s_i^*$. This is because player $j\neq i$ is already best responding to player $i$'s constant action by construction of $\bar{s}_{j}$ and player $j$ cannot create another fixed point by unilaterally deviating from $(s^*_i,\bar{s}_{-i})$ because $s^*_i$ is a constant conditional strategy. Player $j$ cannot deviate to a conditional strategy that admits no fixed point either. Thus, it implies that $(s^*_i,\bar{s}_{-i})$ is a conditional strategy equilibrium.
\end{proof}

\section{Conditional mixed extension}
\label{sec:mixed}

In this section, we give a construction of the conditional extension of a game in mixed extension in an analogous way to the conditional extension of a game in pure strategies. Regardless of the extended utility function, we show that every `conditional mixed strategy' can be induced by a probability measure over `conditional pure strategies'.

Let $G=(\Delta A_i, u_i)_{i\in N}$ be an $n$-person noncooperative game in mixed extension, where $N=\{1,...,n\}$ is the finite set of players, $\Delta A_i$ the set of all probability distributions over the finite action set $A_i$, which is the simplex in $\mathbb{R}^{m_i-1}$, and $u_i:\Delta A\rightarrow \mathbb{R}$ the von Neumann-Morgenstern \citep{neumann1944} expected utility function of player $i\in N$. A mixed strategy profile is denoted by $p\in \Delta A=\times_{i\in N} \Delta A_i$.

Let $\mathcal{A}_{-i}$ be the Borel $\sigma$-algebra over $\Delta A_{-i}$. The set of conditional mixed strategies of player $i$ is given by $\Sigma_i=\{\sigma_i:\Delta A_{-i}\rightarrow \Delta A_i~|~\sigma_i~\text{is a simple and}~\mathcal{A}_{-i}\text{-measurable}\}$. A conditional mixed strategy profile is $\sigma \in \Sigma=\times_{i\in N} \Sigma_{i}$. A conditional strategy profile $\sigma\in \Sigma$ is an \textit{agreement} if there exists a unique $p\in \Delta A$ such that $\sigma(p)=p$---i.e., for all $i$ $\sigma_i(p_{-i})=p_i$---which is the unique fixed point of the conditional strategy profile $\sigma$, which always exists because for each $i$ $\Delta A_i$ is nonempty. The intuition is that everybody agrees on the unique strategy profile $p$. If $\sigma\in \Sigma$ is not an agreement, then it is called a \textit{disagreement}. 

\subsection{Conditional mixed extension and the mixed extension of the  pure conditional strategies}

Let $\hat{S}_i = \{\hat{s}_i:\Delta A_{-i}\rightarrow A_i~|~\hat{s}_i~\text{is}~\mathcal{A}_{-i}\text{-measurable}\}$, which is the set of pure conditional strategies of player $i$ against mixed strategies of the others. We next extend $\hat{S}_i$ to $\Delta \hat{S}_i$, which is the set of all probability measures with finite support over $\hat{S}_i$. Every probability measure $\mu\in \Delta \hat{S}_i$ induces a function $\sigma_i\in \Sigma_{i}$. Define function 
\begin{equation}
\label{eq:phi}
\begin{array}{c c c c}
\phi: & \Delta \hat{S}_i &\rightarrow& \Sigma_{i} \\
& \mu &\mapsto & \sigma_{i}^\mu
\end{array}
\end{equation}
where for all $q_{-i}$ 
\begin{equation*}
\sigma_{i}^\mu (q_{-i}) = \int_{\hat{S}_i} \hat{s}_i (q_{-i}) d\mu(\hat{s}_i),
\end{equation*}
which is a probability distribution over $A_i$ given $q_{-i}$.

Next, we show that every conditional mixed strategy $\sigma_i$ can be induced by a probability measure $\mu$ over conditional (pure) strategies.

\begin{theorem}[Mixed extension]
	\label{thm:conditional_mixed_extension}
	For every $\sigma_{i}\in\Sigma_{i}$ there exists a $\mu\in \Delta \hat{S}_i$ such that $\phi(\mu)=\sigma_{i}$ where $\phi$ is defined in (\ref{eq:phi}).
\end{theorem}

\begin{proof}
	Fix $\sigma_{i}\in\Sigma_{i}$ and consider the induced partition over $\Delta A_{-i}$  $\mathcal{P}_{\sigma_{i}}=\{X_1,...,X_{L_{\sigma_{i}}}\}$. For ease of notation, we denote $\mathcal{P}=\mathcal{P}_{\sigma_{i}}$ and $L=L_{\sigma_{i}}$. Note that $\sigma_{i}$ can be written as
	\[
	\begin{array}{c c c c}
	\sigma_i: &\mathcal{P} &\rightarrow& \Delta A_i \\
	& l &\mapsto & \sigma_{i}(l)=\mu_l
	\end{array}
	\]
	where $\mu_l$ is a probability measure over $A_i$. Put differently, for $a_i\in A_i$, $\mu_l(a_i)=\sigma_{i}(q_{-i})(a_i)$ for all $q_{-i}\in X_l$. The finite collection of probability spaces $(A_i,2^{A_i}, \mu_l)^{L}_{l=1}$ induces the product space $(\times^{L}_{l=1} A_i,2^{\times^{L}_{l=1} A_i}, \mu)$ where $\mu$ is the usual product measure over $2^{\times^{L}_{l=1}A_i}$ defined by  $\mu((a^1_i,a^2_i,...,a^L_i))=\prod_{l=1}^{L} \mu_l(a^l_i)$ for all $(a^1_i,a^2_i,...,a^L_i)$ in  $\times^{L}_{l=1} A_i=\{(a^1_i,a^2_i,...,a^L_i)|a^l_i\in A_i \}$.
	
	Note that every element $(a^1_i,a^2_i,...,a^L_i)\in\times^{L}_{l=1} A_i$ can be identified with a $\hat{s}_i\in \hat{S}_i|_{\mathcal{P}}$ where $\hat{S}_i|_{\mathcal{P}}=\{\hat{s}_i\in \hat{S}_i |\hat{s}_i~\text{is}~\mathcal{P}\text{-measurable}\}$, i.e.,  $\hat{s}_i(q_{-i})=a^l_i$ for all $l$ and all $q_{-i}\in X_l$. Thus, $\mu$ is a measure supported by $\hat{S}_i|_{\mathcal{P}}$ such that for all $\hat{s}_i\in \hat{S}_i|_{\mathcal{P}}$, $\mu(\hat{s}_i)=\prod_{l=1}^{L} \mu_l(\hat{s}_i(l))$ where $\hat{s}_i(l)=\hat{s}_i(q_{-i})$ for all $q_{-i}\in X_l$.
	
	Next we show that $\phi(\mu)=\sigma_{i}$. Recall that by (\ref{eq:phi}) $\phi(\mu)=\sigma^\mu_{i}$, which is defined by for all $q_{-i}$ and all $a_i$ 
	\[
	\sigma_{i}^\mu (q_{-i})(a_i) = \int_{\hat{S}_i} \hat{s}_i (q_{-i})(a_i) d\mu(\hat{s}_i)= \int_{\hat{S}_i|_{\mathcal{P}}} \hat{s}_i (q_{-i})(a_i) d\mu(\hat{s}_i).
	\]
	Note that 
	\[
	\hat{s}_i (q_{-i})(a_i) =
	\begin{cases}
	0~\text{if}~\hat{s}_i (q_{-i}) \neq a_i\\
	1~\text{if}~\hat{s}_i (q_{-i}) = a_i,
	\end{cases}
	\]
	and, therefore,  
	\[
	\sigma_{i}^\mu (q_{-i})(a_i) = 
	\]
	\[
	\int \mathds{1}_{\{\hat{s}_i\in \hat{S}_i|_{\mathcal{P}}:\hat{s}_i(q_{-i})=a_i\}} d\mu(\hat{s}_i) = \mu(\{\hat{s}_i\in \hat{S}_i|_{\mathcal{P}}:\hat{s}_i(q_{-i})=a_i\}).
	\]
	Fix $a_i\in A_i$ and $q_{-i}\in X_{l'}$ for some $l'\in\{1,...,L\}$.  Then, $\{\hat{s}_i\in \hat{S}_i|_{\mathcal{P}}:\hat{s}_i(q_{-i})=a_i\}=\{\times^{l'-1}_{l=1} A_i \times \{a_i\} \times^{L}_{l=l'+1}A_i\}$ and since $\mu$ is the product measure,
	\[
	\mu(\times^{l'-1}_{l=1} A_i \times \{a_i\} \times^{L}_{l=l'+1}A_i) = \mu_{l'}(a_i)=\sigma_i(q_{-i})(a_i).
	\]
	This implies that for all $q_{-i}\in \Delta A_{-i}$ and all $a_i\in A_i$, $\sigma_{i}^\mu (q_{-i})(a_i)=\sigma_i(q_{-i})(a_i)$, i.e., $\phi(\mu)=\sigma_{i}$.
\end{proof}

Theorem~\ref{thm:conditional_mixed_extension} shows that mixed conditional strategies can be obtained by mixing over pure conditional strategies. It is then relatively straight-forward to show that there is always a mixed conditional strategy equilibrium in every $n$-person game, irrespective of the extended utility function. This is because a Nash equilibrium in the form of a constant conditional mixed strategy profile creates a unique agreement and there is no unilateral profitable deviation from it because (i) the agreement constitutes a Nash equilibrium, and (ii) a player cannot create a disagreement by a unilateral  deviation when the other players play constant conditional strategies.

	\bibliographystyle{chicago}
	\bibliography{ref}

\end{document}